\documentclass{article}
\usepackage{amsmath,amsthm,amsfonts,amssymb}

\newtheorem{theorem}{Theorem}

\newtheorem{definition}[theorem]{Definition}
\newtheorem{lemma}[theorem]{Lemma}
\newtheorem{proposition}[theorem]{Proposition}
\newtheorem{problem}[theorem]{Problem}

\def\Prop #1: #2\par{\medbreak \noindent {\bf Proposititon #1 :}\enspace
{\rm #2}\smallskip\noindent}

\def\Df #1: {\medbreak \noindent {\bf Definition #1} :\enspace}

\def\df #1: #2\par{\smallskip\noindent {\bf #1}\enspace{\bf
Defin\'{\i}ci\'o}:\enspace #2\smallskip\noindent}

\def\Al #1: #2\par{\smallskip\noindent {\bf #1}\enspace
{\bf \'All\'{\i}t\'as}:\enspace #2\smallskip\noindent}

\def\C{$C^{\ast}$-}

\def\cL{{\cal L}}

\def\c{{\bot}} 

\def\vagy{\vee}
\def\es{\wedge}

\title{Characterizing common cause closedness \\
of quantum probability theories\footnote{Submitted for publication}}

\author{
Yuichiro Kitajima\\
College of Industrial Technology\\ 
Nihon University\\
2-11-1 Shin-ei, Narashino, Chiba 275-8576, Japan\\
kitajima.yuichirou@nihon-u.ac.jp
\\
\\
Mikl\'os R\'edei\\
Department of Philosophy, Logic and Scientific Method\\
London School of Economics and Political Science\\
Houghton Street, London WC2A 2AE, United Kingdom\\
and\\
Institute of Philosophy \\
Research Center for the Humanities \\
Hungarian Academy of Sciences\\ 
Orsz\'agh\'az utca 30, 1014
Budapest, Hungary\\
m.redei@lse.ac.uk
}

\begin{document}
\maketitle

\abstract{We prove new results on common cause closedness of quantum probability spaces, where by a quantum probability space is meant the projection lattice of a non-commutative von Neumann algebra together with a countably additive probability measure on the lattice. Common cause closedness is the feature that for every correlation between a pair of commuting projections there exists in the lattice a third projection commuting with both of the correlated projections and which is a Reichenbachian common cause of the correlation. The main result we prove is that a quantum probability space is common cause closed if and only if it has at most one measure theoretic atom. This result improves earlier ones published in \cite{GyenisZ-Redei2013Erkenntnis}. The result is discussed from the perspective of status of the Common Cause Principle. Open problems on common cause closedness of  general probability spaces $(\mathcal{L},\phi)$ are formulated, where $\cL$ is an orthomodular bounded lattice and $\phi$ is a probability measure on $\mathcal{L}$.
}


\section{The main result}

In this paper we prove new results on common cause closedness of quantum probability spaces. By a quantum probability space is meant here the projection lattice of a non-commutative von Neumann algebra together with a countably additive probability measure on the lattice. Common cause closedness is the feature that for every correlation between a pair of commuting projections there exists in the lattice a third projection commuting with both of the correlated projections and which is a Reichenbachian common cause of the correlation.

The main result we prove is that a quantum probability space is common cause closed if and only if it has at most one measure theoretic atom. Since classical, Kolmogorovian probability spaces were proved in \cite{GyenisZ-Redei2011} to be common cause closed if and only if they contained at most one measure theoretic atom, and since classical probability spaces also can be regarded as projection lattices of \emph{commutative} von Neumann algebras, this result gives a complete characterization of common cause closedness of probability spaces in the category of von Neumann algebras. Previous results on common cause closedness of quantum probability spaces had to assume an additional, somewhat artificial and not very transparent feature of the quantum probability measure under which the quantum probability space could be proved to be common cause closed \cite{GyenisZ-Redei2013Erkenntnis}. With the removal of that condition it becomes visible that exactly the same type of measure theoretical structure is responsible for the common cause closedness (or lack of it) of classical and quantum probability spaces.

The broader context in which we give our proofs is the problem of characterization of common cause closedness of general probability spaces $(\cL,\phi)$, where $\cL$ is an orthocomplemented, orthomodular, bounded $\sigma$-lattice and $\phi$ is a countably additive general probability measure on $\cL$. (Classical and quantum probability spaces are obviously special examples of abstract probability spaces.) Little is known about the problem of common cause closedness in this generality however. A sufficient condition for common cause closedness of general probability theories is known (Proposition 3.10 in \cite{GyenisZ-Redei2013Erkenntnis}, recalled here as Proposition \ref{Q_closed}) but the condition is exactly the not very natural one that could be eliminated  both in classical and in quantum probability spaces, and one would like to know whether it also can be eliminated (or replaced by a more natural one) in general probability theories (Problem \ref{prob:sufficient}). It also is unknown whether the condition which is necessary for common cause closedness of quantum probability spaces is necessary for the common cause closedness of general probability theories as well (Problem \ref{prob:necessary}). Further open questions and possible directions of investigation will be indicated in section \ref{sec:open}.

The conceptual-philosophical significance of common cause closed probability spaces is that they display a particular form of causal completeness: these theories themselves can explain, exclusively in terms of common causes that they contain, \emph{all} the correlations they predict; hence these theories comply in an extreme manner with the Common Cause Principle. Probabilistic physical theories in which the probability space is measure theoretically purely non-atomic are therefore good candidates for being a confirming evidence for the Common Cause Principle. Section \ref{sec:discussion} discusses this foundational-philosophical significance of the presented results in the context of the more general problem of assessing the status of the Common Cause Principle.

Further sections of the paper are organized as follows.  Section \ref{sec:notations} fixes some notation and recalls some basic definitions in lattice theory. In section \ref{sec:definitionCCC} the notion of common cause in a general probability theories is defined. In Section \ref{sec:sufficient} it is shown that for a probability space, classical or quantum, to be common cause closed it is \emph{sufficient} that they have at most one measure theoretic atom (Propositions \ref{closed_Boolean} and \ref{closed_algebra}). Section \ref{sec:necessary} proves that this condition is also \emph{necessary}, both in case of classical probability spaces (Proposition \ref{characterization_Boolean}) and in quantum probability spaces (Proposition \ref{characterization_algebra}). Section \ref{sec:open} formulates some problems that are open at this time.

\section{General probability spaces –- definitions and notations\label{sec:notations}}

Throughout the paper $\mathcal{L}$ denotes an orthocomplemented lattice with lattice operations $\vagy,\es$ and orthocomplementation $\c$. The lattice $\cL$  is called orthomodular if, for any $A, B \in \mathcal{L}$ such that $A \leq B$, we have
\begin{equation}\label{eq:orthomodular}
B=A \vee (B \wedge A^{\perp})
\end{equation}
The lattice $\mathcal{L}$ is called a Boolean algebra if it is distributive, i.e. if for any $A,B,C \in \mathcal{L}$ we have
\begin{equation}\label{eq:distributive}
A \vee (B \wedge C)=(A \vee B) \wedge (A \vee C)
\end{equation}
It is clear that a Boolean algebra is an orthomodular lattice. Other examples of orthomodular lattices are the lattices of projections of a von Neumann algebra; they are called von Neumann lattices. The projection lattice of a von Neumann algebra is distributive if and only if the von Neumann algebra is commutative. A basic reference for orthocomplemented lattices is \cite{Kalmbach1983}. For a summary of basic facts about von Neumann algebras and von Neumann lattices we refer to \cite{Redei1998}, for the theory of von Neumann algebra our reference is \cite{Kadison-Ringrose1986}. The paper \cite{Redei-Summers2007} gives a concise review of the basics of quantum probability theory.

If, for every subset $S$ of $\mathcal{L}$, the join and the meet of all elements in $S$ exist, then $\mathcal{L}$ is called a complete orthomodular lattice. If the join and meet of all elements of every countable subset $S$ of $\mathcal{L}$ exist in $\mathcal{L}$, then $\mathcal{L}$ is called a $\sigma$-lattice. Von Neumann lattices are complete hence $\sigma$-complete. In the present paper, it is assumed that lattices are bounded: they have a smallest and a largest element denoted by $0$ and $I$, respectively.

Let $\mathcal{L}$ be a $\sigma$-complete orthomodular lattice. Elements $A$ and $B$ are called mutually orthogonal if $A \leq B^{\perp}$. The map $\phi:\mathcal{L} \rightarrow [0,1]$ is called a (general) probability measure if $\phi(I)=1$ and $\phi(A \vee B)=\phi(A)+\phi(B)$ for any mutual orthogonal elements $A$ and $B$.
A probability measure $\phi$ is called a $\sigma$-additive probability measure if for any countable, mutually orthogonal elements $\{ A_i | i \in \mathbb{N} \}$, we have
\begin{equation}
\phi(\vee_{i \in \mathbb{N}} A_i)=\sum_{i \in \mathbb{N}} \phi(A_i)
\end{equation}

Next we consider atoms. There are two types of atoms: algebraic and measure theoretic. An element $A \in \mathcal{L}$ is called an algebraic atom if $A>0$ and, for any $B \leq A$ we have $B=A$ or $B=0$. The other type of atom depends on a probability measure on $\mathcal{L}$. Let $\phi$ be a probability measure on $\mathcal{L}$. An element $A \in \mathcal{L}$ is called a $\phi$-atom if $\phi(A)>0$ and, for any $B \leq A$ we have $\phi(B)=\phi(A)$ or $\phi(B)=0$.

A probability measure $\phi$ on $\mathcal{L}$ is called
\begin{itemize}
\item purely atomic if for any $A \in \mathcal{L}$ with $\phi(A) > 0$ there exists a $\phi$-atom $B \in \mathcal{L}$ such that $B \leq A$,
\item purely nonatomic if for any $A \in \mathcal{L}$ with $\phi(A)>0$ there exists an element $B \in \mathcal{L}$ such that $B < A$ and $0<\phi(B)<\phi(A)$.
\end{itemize}

If $\phi(A)=0$ implies $A=0$ for any $A \in \mathcal{L}$, then $\phi$ is called faithful. Roughly speaking, this condition means that the elements whose probabilities are zero are ignored because such elements are identified with the zero element. According to the following Lemma, we can identify a $\phi$-atom with an atom in the case where $\phi$ is faithful. Since we will deal with faithful measures in the paper, algebraic and measure theoretic atoms can be identified and this will be done implicitly in this paper.

\begin{lemma}
\label{faithful}
Let $\phi$ be a faithful probability measure on $\mathcal{L}$. $A$ is a $\phi$-atom if and only if $A$ is an atom.
\end{lemma}

\begin{proof}
Let $A$ be a $\phi$-atom. For any $B$ such that $B \leq A$, $\phi(B)=\phi(A)$ or $\phi(B)=0$. $\phi(B)=0$ implies $B=0$ because $\phi$ is faithful. $\phi(B)=\phi(A)$ implies $\phi(A^{\perp} \wedge B)=0$, so that $A = B \vee (A^{\perp} \wedge B)=B$. This means that $A$ is an atom. It is trivial that $A$ is a $\phi$-atom if $A$ is an atom.
\end{proof}
We say that two elements $A$ and $B$ in an orthomodular lattice $\mathcal{L}$ are compatible if
\begin{equation}\label{eq:compatible}
A=(A \wedge B) \vee (A \wedge B^{\perp})
\end{equation}
It can be shown \cite{Kalmbach1983}[Theorem 3.2] that (\ref{eq:compatible}) holds if and only if
\begin{equation}
B=(B \wedge A) \vee (B \wedge A^{\perp})
\end{equation}
In other words, the compatibility relation is symmetric. If $\mathcal{L}$ is a Boolean algebra, any two elements in $\mathcal{L}$ are compatible.

\section{Definition of common cause closedness\label{sec:definitionCCC}}

In order to investigate common cause closedness in an orthomodular lattice, we must re-define both the concept of correlation and the notion of common cause of the correlation with which we define common cause closedness. The reason why these concepts have to be re-defined explicitly is that Reichenbach's original definition was given in terms of classical probability spaces \cite{Reichenbach1956}[Section 19], and in such probability spaces all random events are compatible. In the lattice $\cL$ of a general probability space there exist elements however which are not compatible. Hence it must be stipulated explicitly whether we allow (i) incompatible elements to be correlated and (ii) the common causes of correlations to be incompatible with the elements in the correlation. We take a conservative route by disallowing such cases:

Let $A$ and $B$ be compatible elements in $\cL$. We say that $A$ and $B$ are (positively) correlated in state $\phi$ on $\cL$ if
\begin{equation}\label{def:correlation}
\phi(A\es B)>\phi(A)\phi(B)
\end{equation}

\begin{definition}[\cite{Hofer-Redei-Szabo-PCC} Definition 6.1]
{\rm
Let $\mathcal{L}$ be an orthomodular lattice, let $\phi$ be a probability measure on $\mathcal{L}$ and let $A$ and $B$ be elements of $\mathcal{L}$. $C \in \mathcal{L}$ is called a common cause of the correlation (\ref{def:correlation}) if $C$ is compatible with both $A$ and $B$, $0<\phi(C)<1$ and $C$ satisfies the following equations:
\begin{eqnarray}
\label{occ1}
\frac{\phi(A \wedge B \wedge C)}{\phi(C)}&=&\frac{\phi(A \wedge C)}{\phi(C)}\frac{\phi(B \wedge C)}{\phi(C)}\\
\label{occ2}
\frac{\phi(A \wedge B \wedge C^{\perp})}{\phi(C^{\perp})}&=&\frac{\phi(A \wedge C^{\perp})}{\phi(C^{\perp})}\frac{\phi(B \wedge C^{\perp})}{\phi(C^{\perp})}\\
\label{occ3}
\frac{\phi(A \wedge C)}{\phi(C)} &>& \frac{\phi(A \wedge C^{\perp})}{\phi(C^{\perp})}\\
\label{occ4}
\frac{\phi(B \wedge C)}{\phi(C)} &>& \frac{\phi(B \wedge C^{\perp})}{\phi(C^{\perp})}
\end{eqnarray}
If $C$ is a common cause of the correlation between $A$ and $B$, and differs from both $A$ and $B$, then we call the element $C$ a nontrivial common cause.
}
\end{definition}
This definition is the complete analogue of Reichenbach's: it coincides with Reichenbach's original definition when one views it in the classical probability space consisting of the distributive sublattice of $\cL$ generated by the elements $A,B$ and $C$ with the restriction of $\phi$ to this Boolean algebra.

We can now give the central definition of the paper:
\begin{definition}
{\rm
Let $\mathcal{L}$ be an orthomodular lattice and let $\phi$ be a probability measure on $\mathcal{L}$. We say that $(\mathcal{L}, \phi)$ is common cause closed if, for any distinct, compatible elements $A$ and $B$ such that  $\phi(A \wedge B)>\phi(A)\phi(B)$, there is a nontrivial common cause $C$ in $\cL$ of the correlation.
}
\end{definition}
Note that one can argue for a more liberal definition of common cause in which the common cause is not required to be compatible with the correlating events. Common cause closedness can be investigated under such a definition of common cause as well but we do not wish to pursue this direction in this paper, see the papers \cite{Hofer1997}, \cite{Hofer1998}, \cite{Hofer-Vecsernyes2012b} and \cite{Hofer-Redei-Szabo-PCC}[Chapter 8] for the notion of non-commutative common cause and some of its features.

\section{A sufficient condition for common cause closedness\label{sec:sufficient}}

Gyenis and R\'{e}dei \cite{GyenisZ-Redei2011} gave a characterization of common cause closedness of classical probability measure spaces. In the proof of that characterization the following result of Johnson \cite{Johnson1970} played an important role. Let $(\mathcal{S}, p)$ be a classical probability measure space. Johnson \cite{Johnson1970} showed that the probability measure $p$ can be decomposed into probability measures $p_1$ and $p_2$,
\begin{equation}
p=\alpha p_1 + (1-\alpha) p_2
\end{equation}
for some $\alpha \in [0,1]$, where $p_1$ is purely atomic and $p_2$ is purely nonatomic.

A similar decomposition is not known to hold in a general probability space. Thus, in order to prove a sufficient condition for common cause closedness of general probability theories, Gyenis and R\'{e}dei \cite{GyenisZ-Redei2013Erkenntnis} introduced the notion of $Q$-decomposability: It is a special case of a decomposition of a probability measure.

\begin{definition}[\cite{GyenisZ-Redei2013Erkenntnis} p. 445]
{\rm
A $\sigma$-additive probability measure $\phi$ on a $\sigma$-complete orthomodular lattice $\mathcal{L}$ is said to be $Q$-decomposable, where $Q$ is an element of $\cL$, if there exists a decomposition $\phi=\alpha \phi_1 + (1-\alpha) \phi_2$, with $\alpha \in [0,1]$, furthermore $\phi_2$ is a purely nonatomic, $\sigma$-additive probability measure, and for $\phi_1$ we have:
\begin{equation}
\phi_1(A) = \begin{cases}
    1 & (\text{if} \ Q \leq A) \\
    0 & (\text{otherwise}).
  \end{cases}
\end{equation}
}
\end{definition}
Making use of $Q$-decomposability, Gyenis and R\'{e}dei \cite{GyenisZ-Redei2013Erkenntnis} strengthened Proposition 3.9 in \cite{Kitajima2008} in the following way:
\begin{proposition}[\cite{GyenisZ-Redei2013Erkenntnis} Proposition 3.10]
\label{Q_closed}
Let $\mathcal{L}$ be a $\sigma$-complete orthomodular lattice and let $\phi$ be a faithful $\sigma$-additive probability measure on $\mathcal{L}$. If there is at most one $\phi$-atom $Q$ in $\cL$, and $\phi$ is $Q$-decomposable, then $(\mathcal{L}, \phi)$ is common cause closed.
\end{proposition}

$Q$-decomposability is a somewhat artificial and non-transparent condition however. Therefore it is desirable to investigate whether one can remove or replace it by another, more natural condition. It is not known if this can be done in general probability theories. We claim however that it can be both in classical and in quantum probability theories. First we examine the classical case. We prove the following lemma to remove the condition of $Q$-decomposability in Proposition \ref{Q_closed}.

\begin{lemma}
\label{one_atom_Boolean}
Let $\mathcal{L}$ be a $\sigma$-complete Boolean algebra and let $\phi$ be a faithful $\sigma$-additive probability measure on $\mathcal{L}$. If $\mathcal{L}$ has at most one $\phi$-atom $Q$, then $\phi$ is $Q$-decomposable.
\end{lemma}

\begin{proof}
If there is no $\phi$-atom in $\cL$, then $\phi$ is purely nonatomic. Thus $\phi$ is $Q$-decomposable.

Suppose that $\mathcal{L}$ has one single $\phi$-atom $Q$. For any $A \in \mathcal{L}$, let the maps $\phi_{1}'$ and $\phi_{2}'$ be defined by
\begin{eqnarray}
\phi_{1}'(A) &:=& \phi(A \wedge Q) \\
\phi_{2}'(A) &:=& \phi(A \wedge (A \wedge Q)^{\perp})
\end{eqnarray}
Then $\phi=\phi_{1}'+\phi_{2}'$ since $A=(A \wedge Q) \vee (A \wedge (A \wedge Q)^{\perp})$. If $\phi_{1}'(I)=1$, then $Q=I$. In this case, $\phi$ is $Q$-decomposable. We, thus, assume that $\phi_{1}'(I)<1$. Then $\phi_{2}'(I)>0$. Define $\phi_1$ and $\phi_2$ as follows:
\begin{eqnarray}
\phi_{1}&:=&(1/\phi_{1}'(I))\phi_{1}'\\
\phi_{2}&:=&(1/\phi_{2}'(I))\phi_{2}'
\end{eqnarray}
and let $\alpha := \phi_{1}'(I)$. Then $\phi=\alpha \phi_1 + (1-\alpha) \phi_2$ and $\phi_1(I)=\phi_2(I)=1$.

For any $A$ such that $Q \leq A$, we have $\phi_{1}(A)=1$. For any $A$ such that $Q \not\leq A$, $A \wedge Q=0$ since $A \wedge Q < Q$ and $Q$ is an atom. Thus $\phi_{1}(A)=0$ for any $A$ such that $Q \not\leq A$.

Next we show that $\phi_1$ and $\phi_2$ are $\sigma$-additive probability measures.

Let $\{ A_i | i \in \mathbb{N} \}$ be a countable set of mutually orthogonal elements in $\mathcal{L}$. We distinguish two cases:
\begin{itemize}
\item[(i)] Suppose that $A_i \wedge Q = 0$ for any $i \in \mathbb{N}$. \\
Because $\mathcal{L}$ is a Boolean algebra, for any $i \in \mathbb{N}$ we have.
\begin{equation}
Q=(A_i \wedge Q) \vee (A_i^{\perp} \wedge Q)=A_i^{\perp} \wedge Q
\end{equation}
This implies $A_i \leq Q^{\perp}$ for any $i \in \mathbb{N}$. Thus, $\vee_{i \in \mathbb{N}} A_i \leq Q^{\perp}$. Therefore
\begin{equation}
(\vee_{i \in \mathbb{N}} A_i) \wedge Q =0
\end{equation}
Hence
\begin{equation}
\phi_1(\vee_{i \in \mathbb{N}} A_i)=\sum_{i \in \mathbb{N}} \phi_1(A_i)=0
\end{equation}
and
\begin{eqnarray}
\phi_2(\vee_{i \in \mathbb{N}} A_i)&=&(1/(1-\alpha))(\phi(\vee_{i \in \mathbb{N}} A_i)-\alpha \phi_1(\vee_{i \in \mathbb{N}} A_i))\\
&=&(1/(1-\alpha))\phi(\vee_{i \in \mathbb{N}} A_i)\\
&=&(1/(1-\alpha))\sum_{i \in \mathbb{N}}\phi(A_i)\\
&=&\sum_{i \in \mathbb{N}}((1/(1-\alpha))\phi(A_i)-\alpha \phi_1(A_i))\\
&=&\sum_{i \in \mathbb{N}}\phi_2(A_i)
\end{eqnarray}

\item[(ii)] Suppose that there is $A_j$ such that $A_j \wedge Q \neq 0$. \\
Since $0 < A_j \wedge Q \leq Q$ and $Q$ is a $\phi$-atom, we have $Q=A_j \wedge Q$. Therefore
\begin{equation}
\phi_1(A_j)=\phi(A_j \wedge Q)/\phi(Q)=1
\end{equation}
and
\begin{eqnarray}
1 \leq \phi_1(A_j)&=&\phi(A_j \wedge Q)/\phi(Q) \leq \phi((\vee_{i \in \mathbb{N}} A_i) \wedge Q)/\phi(Q)\\
&=&\phi_1(\vee_{i \in \mathbb{N}} A_i) \leq 1
\end{eqnarray}
Moreover, $A_i^{\perp} \geq A_j \geq Q$ for any $i \in \mathbb{N}$ such that $i \neq j$, so that $A_i \wedge Q = 0$. Thus $\phi_1(A_i)=0$ for any $i \in \mathbb{N}$ such that $i \neq j$. Hence
\begin{equation}
\phi_1(\vee_{i \in \mathbb{N}} A_i)=\sum_{i \in \mathbb{N}} \phi_1(A_i)=1
\end{equation}
and
\begin{eqnarray}
\phi_2(\vee_{i \in \mathbb{N}} A_i)&=&(1/(1-\alpha))(\phi(\vee_{i \in \mathbb{N}} A_i)-\alpha \phi_1(\vee_{i \in \mathbb{N}} A_i))\\
&=&(1/(1-\alpha))(\phi(\vee_{i \in \mathbb{N}} A_i)-\alpha)\\
&=&(1/(1-\alpha))(\phi(A_j)-\alpha)+(1/(1-\alpha))\sum_{i \neq j}\phi(A_i)\\
&=&(1/(1-\alpha))(\phi(A_j)-\alpha \phi_1(A_j))+(1/(1-\alpha))\\
&{}&\times \sum_{i \neq j}(\phi(A_i)-\alpha \phi_1(A_i))\\
&=&\phi_2(A_j)+\sum_{i \neq j}\phi_2(A_i)\\
&=&\sum_{i \in \mathbb{N}}\phi_2(A_i)
\end{eqnarray}
\end{itemize}
Finally we show that $\phi_2$ is purely nonatomic. Let $A$ be an element in $\mathcal{L}$ such that $\phi_2(A)>0$. Then $A \wedge (A \wedge Q)^{\perp} > 0$ since $\phi$ is faithful. Suppose that $Q \leq A \wedge (A \wedge Q)^{\perp}$. Then $Q \leq A$ and $Q \leq (A \wedge Q)^{\perp}$, which implies $Q=A \wedge Q \leq Q^{\perp}$. This is a contradiction. Thus $Q \not\leq A \wedge (A \wedge Q)^{\perp}$; that is we have
\begin{equation}
\label{Q_atom}
Q \wedge (A \wedge (A \wedge Q)^{\perp})=0
\end{equation}
Because $Q$ is the only $\phi$-atom by assumption, $A \wedge (A \wedge Q)^{\perp}$ is not a $\phi$-atom. So there is an element $X \in \mathcal{L}$ such that
\begin{equation}
0<X<A \wedge (A \wedge Q)^{\perp}
\end{equation}
By Equation (\ref{Q_atom}) we have $X \wedge Q=0$. Thus
\begin{equation}
X=X \wedge (X \wedge Q)^{\perp}
\end{equation}
Hence
\begin{equation}
0 < \phi(X)=\phi(X \wedge (X \wedge Q)^{\perp})=\phi_2'(X)
\end{equation}
and
\begin{equation}
\phi_2'(X)=\phi(X)<\phi(A \wedge (A \wedge Q)^{\perp})=\phi_2'(A)
\end{equation}
This shows that $\phi_2$ is purely nonatomic.
\end{proof}

By Proposition \ref{Q_closed} and Lemma \ref{one_atom_Boolean}, we get the following proposition as a corollary.

\begin{proposition}
\label{closed_Boolean}
Let $\mathcal{L}$ be a $\sigma$-complete Boolean algebra and let $\phi$ be a faithful $\sigma$-additive probability measure on $\mathcal{L}$. If $\mathcal{L}$ has at most one $\phi$-atom, then $(\mathcal{L}, \phi)$ is common cause closed.
\end{proposition}

Next we consider the case of quantum probability spaces. We proceed along the logic followed in the classical case: we prove a lemma that is analogous to Lemma \ref{one_atom_Boolean} showing that $Q$-decomposability is entailed by the feature that $Q$ is the only measure theoretic atom:
\begin{lemma}
\label{one_atom_algebra}
Let $\mathcal{L}$ be the orthomodular lattice of projections of a von Neumann algebra $\mathfrak{N}$, and let $\phi$ be a faithful normal state on $\mathfrak{N}$. If $\mathcal{L}$ has at most one $\phi$-atom $Q$, then $\phi$ is $Q$-decomposable.
\end{lemma}

\begin{proof}
Let $\mathfrak{N}$ be a von Neumann algebra, and let $\mathfrak{C}(\mathfrak{N})$ be the center of $\mathfrak{N}$. There are projections $P_{I},P_{n} \in \mathfrak{C}(\mathfrak{N})$ such that $P_{I}+P_{n}=I$, $\mathfrak{N}P_{I}$ is of type I and $\mathfrak{N}P_{n}$ is not of type I (Theorem 6.5.2 in \cite{Kadison-Ringrose1986}).

If there is no $\phi$-atom in $\cL$, then $P_I=0$. In this case, $\phi$ is purely nonatomic. Thus, $\phi$ is $Q$-decomposable. If $P_I=I$, then $P_I=Q$. Thus $\phi$ is $Q$-decomposable.

Suppose that $0<P_I<I$. Define $\phi_1$ and $\phi_2$  by
\begin{eqnarray}
\phi_1(A)&:=&\phi(AP_I)/\phi(P_I)\\
\phi_2(A)&:=&\phi(AP_n)/\phi(P_n)
\end{eqnarray}
and let $\alpha := \phi(P_I)$. Then $\phi=\alpha \phi_1 + (1- \alpha) \phi_2$.

Suppose that $\mathfrak{N}P_I$ contains two projections. Then it contains two $\phi$-atoms because $\mathfrak{N}P_I$ is of type I. This is a contradiction. Thus, $\mathfrak{N}P_I$ contains only one projection. Therefore, $P_I=Q$.

For any projection $A \in \mathfrak{N}$ such that $Q \leq A$, we have
\begin{equation}
\phi_1(A)=\phi(AP_I)/\phi(P_I)=\phi(AQ)/\phi(Q)=1
\end{equation}
For any projection $A \in \mathfrak{N}$ such that $Q \not\leq A$, $A$ is in $\mathfrak{N}P_n$, that is, $A \leq P_n$. Thus, $\phi_1(A)=\phi(AP_I)/\phi(P_I)=0$. $\mathfrak{N}P_n$ does not have any $\phi$-atom, $\phi_2$ is purely nonatomic.
\end{proof}
By Proposition \ref{Q_closed} and Lemma \ref{one_atom_algebra}, we get the following proposition as a corollary.
\begin{proposition}
\label{closed_algebra}
Let $\mathcal{L}$ be an orthomodular lattice of projections of a von Neumann algebra $\mathfrak{N}$ and let $\phi$ be a faithful $\sigma$-additive probability measure on $\mathcal{L}$. If $\mathcal{L}$ has at most one $\phi$-atom, then $(\mathcal{L}, \phi)$ is common cause closed.
\end{proposition}

Propositions \ref{closed_Boolean} and \ref{closed_algebra} mean that the condition that $\mathcal{L}$ has at most one $\phi$-atom is a sufficient condition for common cause closedness in the category of von Neumann algebras.

\section{A necessary condition for common cause closedness\label{sec:necessary}}
Next we consider the necessary condition for common cause closedness. Gyenis and R\'{e}dei \cite{GyenisZ-Redei2013Erkenntnis} showed that an orthomodular lattice which has two distinct $\phi$-atoms is not common cause closed under an additional condition:

\begin{proposition}[\cite{GyenisZ-Redei2013Erkenntnis} Proposition 3.8]
\label{Gyenis_not_common_cause_closed}
Let $\mathcal{L}$ be a $\sigma$-complete orthomodular lattice and let $\phi$ be a faithful $\sigma$-additive probability measure on $\mathcal{L}$. If $\mathcal{L}$ contains two distinct $\phi$-atoms $P$ and $Q$ such that $\phi(P \vee Q) < 1$, then $(\mathcal{L}, \phi)$ is not common cause closed.
\end{proposition}

We show that if $(\cL,\phi)$ is either a classical or a quantum probability space, the assumption of condition $\phi(P \vee Q) < 1$ in Proposition \ref{Gyenis_not_common_cause_closed} can be replaced by the condition that there is a pair of correlated elements in $\cL$. This latter condition is more natural to assume because the aim of the Reichenbachian common cause is to explain correlations.

First we examine the case of a classical probability space.

\begin{lemma}
\label{correlation_Boolean}
Let $\mathcal{L}$ be a Boolean algebra and let $\phi$ be a faithful probability measure on $\mathcal{L}$. Assume that $\cL$ contains two atoms $P$ and $Q$ (recall that atoms are also $\phi$-atoms by faithfulness of $\phi$). There are two distinct elements $A$ and $B$ in $\mathcal{L}$ such that $\phi(A \wedge B)>\phi(A)\phi(B)$ if and only if $\phi(P \vee Q)<1$.
\end{lemma}

\begin{proof}
First we prove that if $A$ and $B$ are elements in $\mathcal{L}$ such that $\phi(A \wedge B)>\phi(A)\phi(B)$ and $P$ and $Q$ are distinct $\phi$-atoms, then $\phi(P \vee Q)<1$. The proof is indirect: Suppose that $\phi(P \vee Q)=1$. We show that this leads to contradiction.

Define $\phi_1$ and $\phi_2$ by
\begin{eqnarray}
\phi_{1}(X)&:=&\phi((X \wedge P) \vee (X \wedge Q))\\
\phi_{2}(X)&:=&\phi(X \wedge \{(X \wedge P) \vee (X \wedge Q)\}^{\perp})
\end{eqnarray}
for any $X \in \mathcal{L}$. Then $\phi=\phi_{1}+\phi_{2}$ since
\begin{equation}
X=\{ (X \wedge P) \vee (X \wedge Q)\} \vee (X \wedge \{(X \wedge P) \vee (X \wedge Q)\}^{\perp}))
\end{equation}
For any $X \in \mathcal{L}$ we have
\[ \begin{split}
\phi_{1}(X \vee X^{\perp}) &=\phi((X \vee X^{\perp}) \wedge P) \vee ((X \vee X^{\perp}) \wedge Q)) \\
&=\phi((X \wedge P) \vee (X^{\perp} \wedge P) \vee (X \wedge Q) \vee (X^{\perp} \wedge  Q)) \\
&=\phi((X \wedge P) \vee (X \wedge Q))+ \phi((X^{\perp} \wedge P) \vee (X^{\perp} \wedge  Q)) \\
&=\phi_{1}(X)+\phi_{1}(X^{\perp})
\end{split} \]
Thus for any $X \in \mathcal{L}$ we have:
\begin{eqnarray}
\phi_{2}(I)&=&\phi_{2}(X \vee X^{\perp})\label{eq:help1}\\
&=&\phi(X \vee X^{\perp})-\phi_{1}(X \vee X^{\perp})\\
&=&\phi(X)-\phi_1(X)+\phi(X^{\perp})-\phi_1(X^{\perp})\\
&=&\phi_{2}(X)+\phi_{2}(X^{\perp}) \\
&\geq& \phi_2(X)\label{eq:help2}
\end{eqnarray}
Since $\phi(P \vee Q)=1$ by our indirect assumption, using the decomposition of $\phi$ into $\phi_1$ and $\phi_2$, and keeping in mind the definition of $\phi_1$ we have
\begin{eqnarray}
\phi_2(I)&=&\phi(I)-\phi_1(I)\\
&=&1-\phi(P\vagy Q)=0
\end{eqnarray}
This, by equations (\ref{eq:help1})-(\ref{eq:help2}), entails $\phi_{2}(X)=0$ for any $X \in \mathcal{L}$. It follows that
\begin{equation}
\phi(X)=\phi_1(X)=\phi((X \wedge P) \vee (X \wedge Q))
\end{equation}
Since $P$ and $Q$ are distinct atoms, one also has:
\begin{eqnarray}
P&=&P \wedge (Q \vee Q^{\perp})\\
&=&(P \wedge Q) \vee (P \wedge Q^{\perp})\\
&=&P \wedge Q^{\perp}\\
& \leq & Q^{\perp}
\end{eqnarray}
It follows that $X \wedge P$ is orthogonal to $X \wedge Q$. Thus for all $X\in \cL$
\begin{equation}
\label{boolean}
\phi(X)=\phi(X \wedge P)+ \phi(X \wedge Q)
\end{equation}
Since $\phi(A \wedge B)>\phi(A)\phi(B)$, we have $\phi(A)>0$ and $\phi(B)>0$. Inserting $X=A$ and $X=B$ into equation (\ref{boolean}), we see that we have the following cases:
\begin{itemize}
\item[(i)] $\phi(A \wedge P)>0$ or $\phi(A \wedge Q)>0$\\
and
\item[(ii)] $\phi(B \wedge P)>0$ or $\phi(B \wedge Q)>0$
\end{itemize}

Suppose that $\phi(A \wedge P)>0$ and $\phi(B \wedge Q)>0$. (The other combinations of the possibilities in (i) and (ii) can be handled exactly the same way.) Then $0 < A \wedge P \leq P$ and $0< B \wedge Q \leq Q$. Since $P$ and $Q$ are atoms, $A \wedge P = P$ and $B \wedge Q = Q$, which entails $P \leq A$ and $Q \leq B$. We now separate the cases according to how $P$ and $A$ and $Q$ and $B$ are related:
\begin{itemize}
\item if $P=A$ and $Q=B$, then $A \wedge B=P \wedge Q=0$ because $P$ and $Q$ are distinct atoms. Hence $0=\phi(A \wedge B)<\phi(A)\phi(B)$. This is a contradiction because $A$ and $B$ were assumed to be positively correlated in $\phi$.
\item If $P<A$ and $Q=B$, then $P^{\perp} \wedge A > 0$ since $A=P \vee (P^{\perp} \wedge A)$. Since $\phi$ is faithful, we have then $\phi(P^{\perp} \wedge A) > 0$. Inserting $X=P^\c  \wedge A$ into quation (\ref{boolean}) we obtain
\begin{equation}
0<\phi(P^\c  \wedge A)=\phi((P^\c  \wedge A) \wedge P)+ \phi((P^\c  \wedge A) \wedge Q)
\end{equation}
which entails $\phi(A \wedge P^{\perp} \wedge Q)>0$. Thus
\begin{equation}\label{eq:help3}
0 < A \wedge P^{\perp} \wedge Q \leq Q
\end{equation}
Since $Q$ is an atom, equation (\ref{eq:help3}) entails
\begin{equation}
A\es P^\c \es Q=Q
\end{equation}
from which $A\es P^\c\geq Q$ follows. Thus $Q \leq A$. The inequalities $P \leq A$ and $Q \leq A$ entail $P\vagy Q\leq A$, thus $1=\phi(P \vee Q) \leq \phi(A)$. Hence $\phi(A \wedge B) \leq \phi(A)\phi(B)$. This contradicts the indirect assumption that $A$ and $B$ are positively correlated in $\phi$.
\item If $P<A$ and $Q<B$, a contradiction follows in a similar way; we leave out the details.
\end{itemize}
Therefore $\phi(P \vee Q)<1$.

We now show that if $P$ and $Q$ are distinct $\phi$-atoms in $\cL$ such that $\phi(P \vee Q)<1$, then there are elements $A$ and $B$ in $\cL$ such that $\phi(A \wedge B)>\phi(A)\phi(B)$. Let $A := P \vee Q$ and $B := P$. Then
\begin{eqnarray}
\phi(A \wedge B)&=&\phi((P \vee Q) \wedge P)\\
&=&\phi(P)\\
&>&\phi(P\vee Q)\phi(P)\\
&=&\phi(A)\phi(B)
\end{eqnarray}
\end{proof}

By Propositions \ref{closed_Boolean} and \ref{Gyenis_not_common_cause_closed} and Lemma \ref{correlation_Boolean}, we get the following proposition.

\begin{proposition}[cf. \cite{GyenisZ-Redei2011} Theorem 1]
\label{characterization_Boolean}
Let $\mathcal{L}$ be a $\sigma$-complete Boolean algebra, let $\phi$ be a faithful $\sigma$-additive probability measure on $\mathcal{L}$ and let $\mathcal{L}$ contain two distinct elements $A$ and $B$ such that $\phi(A \wedge B) > \phi(A)\phi(B)$.
$(\mathcal{L}, \phi)$ is common cause closed if and only if $\mathcal{L}$ has at most one $\phi$-atom.
\end{proposition}

Next we consider quantum probability spaces. In this case, a result similar to Lemma  \ref{correlation_Boolean} holds.

\begin{lemma}
\label{correlation_algebra}
Let $\mathcal{L}$ be the orthomodular lattice of projections of a von Neumann algebra $\mathfrak{N}$ and let $\phi$ be a faithful normal state of $\mathfrak{N}$. Assume $\mathcal{L}$ contains two atoms $P$ and $Q$ (which are $\phi$-atoms as well by Lemma \ref{faithful}). There are two distinct elements $A$ and $B$ in $\mathcal{L}$ such that $\phi(A \wedge B)>\phi(A)\phi(B)$ if and only if $\phi(P \vee Q)<1$.
\end{lemma}

\begin{proof}
We prove first that if the von Neumann lattice $\mathcal{L}$ has two distinct atoms $P$ and $Q$ and two projections $A$ and $B$ such that $\phi(A \wedge B)>\phi(A)\phi(B)$, then $\phi(P \vee  Q)<1$.

Let $\mathfrak{C}(\mathfrak{N})$ be the center of $\mathfrak{N}$. By Theorem 6.5.2 in \cite{Kadison-Ringrose1986} there are projections $P_{I},P_{n} \in \mathfrak{C}(\mathfrak{N})$ such that
\begin{itemize}
\item $P_{I}+P_{n}=I$
\item $\mathfrak{N}P_{I}$ is of type I
\item $\mathfrak{N}P_{n}$ is \emph{not} of type I
\end{itemize}
We again separate the cases:
\begin{itemize}
\item[(a)] Suppose that $P_{n} > 0$. If $P$ and $Q$ are distinct atoms, then $\phi(P \vee  Q)<1$ since $P,Q \leq P_{I} < P_{I}+P_{n}=I$.

\item[(b)] Let $P_{n}=0$. Then $\mathfrak{N}$ is of type I. \\
Subcases:
\begin{itemize}
\item[(i)] Suppose that $\mathfrak{N}$ is of type $\text{I}_{2}$. \\
Since $\phi(A \wedge B)>\phi(A)\phi(B)$, we have $\phi(A)>0$ and $\phi(B)>0$. Thus there are atoms $R_{1}$ and $R_{2}$ such that $R_{1} \leq A$ and $R_{2} \leq B$. If $R_1=A$ and $R_2=B$, then $A \wedge B=R_1 \wedge R_2=0$. Hence $\phi(A \wedge B) < \phi(A)\phi(B)$. This is a contradiction. If $R_1 < A$, then $0 < A \wedge R_1^{\perp}$. Thus there is an atom $R_3$ such that $R_3 \leq A \wedge R_1^{\perp}$. Since $\mathfrak{N}$ is of type $\text{I}_2$, $R_1 \vee R_3 = I$, which implies $I = R_1 \vee R_3 \leq A$. Hence $\phi(A \wedge B) \leq \phi(A)\phi(B)$. This is a contradiction. A contradiction follows similarly in the case where $R_2 <B$. Hence $\mathfrak{N}$ is not of type $\text{I}_{2}$.
\item[(ii)] Suppose $\mathfrak{N}$ is of type $\text{I}_1$. \\
Then there are two distinct atoms $P$ and $Q$ such that $\phi(P \vee Q)<1$ by Lemma \ref{correlation_Boolean} because the von Neumann algebra $\mathfrak{N}$ is abelian in this case.
\item[(iii)] Suppose $\mathfrak{N}$ is of type $\text{I}_n$, where $n \geq 3$.\\
Then there are three mutually orthogonal projections $P$, $Q$, and $R$. Thus $\phi(P \vee Q)<1$.
    \end{itemize}
\end{itemize}

The proof that if there are two atoms $P$ and $Q$ in von Neumann lattice $\mathcal{L}$ such that $\phi(P \vee Q)<1$ then there are elements $A,B$ in $\mathcal{L}$ that are positively correlated in $\phi$ is the same as the corresponding proof
in Lemma \ref{correlation_Boolean}.
\end{proof}

By Propositions \ref{closed_algebra} and \ref{Gyenis_not_common_cause_closed} and Lemma \ref{correlation_algebra}, we get the following proposition.

\begin{proposition}
\label{characterization_algebra}
Let $\mathcal{L}$ be an orthomodular lattice of projections of a von Neumann algebra $\mathfrak{N}$, let $\phi$ be a faithful normal state of $\mathfrak{N}$ and let $\mathfrak{N}$ contain two distinct projections $A$ and $B$ such that $\phi(A \wedge B) > \phi(A)\phi(B)$. $(\mathcal{L}, \phi)$ is common cause closed if and only if $\mathcal{L}$ has at most one $\phi$-atom.
\end{proposition}

\section{Common cause closedness and the Common Cause Principle\label{sec:discussion}}

The philosophical significance of common cause closedness of probability spaces is that they display a particular form of causal completeness: these theories themselves can explain \emph{all} the correlations they predict, and the explanation is \emph{exclusively} in terms of common causes they contain. Hence these theories comply in an extreme manner with the Common Cause Principle.

Recall that the Common Cause Principle states that a correlation between events is either due to a direct causal link between the correlated events, or there is a third event, a common cause that explains the correlation. This principle, which goes back to Reichenbach \cite{Reichenbach1956}, and which was sharply articulated mainly by Salmon \cite{Salmon1978}, \cite{Salmon1984}, is a strong claim about the relation of correlations and causality, and its status has been subject of intense analysis in philosophy of science.

The initial phase of the investigations in works of Salmon and others such as van Fraassen \cite{Fraassen1982b}, Cartwright \cite{Cartwright1987}, Sober \cite{Sober1988}, \cite{Sober1989}, \cite{Sober2001} was semi-formal: precise models of the phenomena the papers analyzed were typically not given in terms of explicitly specified probability measure spaces. Thus the reasoning remained vague in crucial respects. The recent trend in investigating the Common Cause Principle is different: The concepts and claims related to the principle get a mathematically precise, sharp definition, they are analyzed mathematically and the results are interpreted from the perspective of their metaphysical-philosophical ramification. Examples of this approach are especially the works of the ``Cracow school", of the ``Bern school" and of the ``Budapest school" (see the books \cite{Wronski2014}, \cite{Hofer-Redei-Szabo-PCC}, \cite{Wuethrich2004}, the references therein, and also the recent papers by Sober and Steel \cite{Sober-Steel2012}, Marczyk and Wronski \cite{Marczyk-Wronski2013a}), Grasshoff et al. \cite{Grasshoff-Portmann-Wuthrich2005}, and Portmann and W\"uethrich \cite{Portmann-Wuthrich2007}.

There have been two broad and important consequences of the more technical analysis of the Common Cause Principle: One is the insight that assessing the status of the Common Cause Principle is a much more complicated and subtle matter than previously thought. The other is that the mathematical analysis has established direct contacts to those scientific theories that have empirically confirmed claims which are directly relevant for the truth (or otherwise) of the Common Cause Principle -- first and foremost to non-relativistic quantum mechanics of finite degrees of freedom and to relativistic quantum field theory.

The philosophical reason why assessing the status of the Common Cause Principle is difficult, is well known from the history of philosophy: the Common Cause Principle makes \emph{two} general existential claims, and we know that falsifying empirically a hypothesis making even \emph{one} existential claim is problematic because one cannot search the whole universe, past, present and future to find a falsifying instance. This was most unambiguously stated by Popper \cite{Popper1959}. In short: Principles such as the Common Cause Principle are of metaphysical character. History of philosophy, from Kant and Hume through Popper teaches us that metaphysical claims cannot be verified or even falsified conclusively.

This general difficulty takes a very specific form in connection with the Common Cause Principle: It is certainly not enough for a falsification of the Common Cause Principle to display a probabilistic theory that contains a correlation but no common cause for it because the hypothetical common cause may be hidden: it may not be part of the theory predicting the correlation, but it may very well be part of a larger, more comprehensive theory having a more detailed picture of the world by using a richer Boolean algebra in its probabilistic part. We may never find that larger theory but we know that its existence is logically possible: every classical probability theory that is not common cause closed can be embedded into a larger one which is; this is one of the conclusions of the formal analysis in  \cite{Hofer-Redei-Szabo1999}, \cite{GyenisZ-Redei2011} (also see \cite{Hofer-Redei-Szabo-PCC}[Chapter 6] and \cite{Marczyk-Wronski2013a}). This feature of classical probability spaces is called common cause cause completability of not common cause closed probability spaces.

Thus, rather than trying to verify or falsify the Common Cause Principle, one should settle on a more modest goal: to look at our best, confirmed scientific theories to see whether they offer evidence in favor or against the truth of the Common Cause Principle (for a more detailed discussions of this point see \cite{Hofer-Redei-Szabo-PCC}[Chapter 10] and \cite{Redei2014CCPAssessment}). It is in connection with such an evaluation of the Common Cause Principle that common cause closedness becomes relevant: If a probabilistic theory is such that the probability theory it employs to describe phenomena is common cause closed then that theory is a strong candidate for being a confirming evidence in favor of the Common Cause Principle.

Crucial in this respect is (relativistic) quantum field theory, which, by its very construction, is intended to comply fully with causality. Since quantum field theory, at least in its mathematically precise, ``axiomatic" form \cite{Haag1992}, \cite{Araki1999}, is based on quantum probability theory, characterizing common cause closedness of such quantum probability theories is relevant from the perspective of evaluating the causal behavior of quantum field theory. The characterization in terms of the measure theoretic atomicity given in this paper makes contact to a deep structural property of quantum field theory: The atomicity structure of quantum probability theories is intimately linked to the Murray-von Neumann classification of (factor) von Neumann algebras: Type $II_1$ and type $III$ factor von Neumann algebras with faithful states define purely non-atomic quantum probability theories \cite{Redei-Summers2007}. Furthermore, it is a fundamental feature of relativistic quantum field theories formulated in terms of covariant local nets of von Neumann algebras that the algebras representing strictly local observables are type $III$ factors \cite{Haag1992}[Section V.6]. Thus these algebras typically define measure theoretically purely non-atomic quantum probability theories. Consequently, relativistic quantum field theories are common cause closed and hence very good candidates for being causally complete in the sense of being capable of explaining causally all the correlations they predict.

Note the careful wording of the previous sentence: quantum field theory has not been proved to be in full compliance with locality. This is because in this framework one can impose locality conditions on the common cause, in addition to the defining conditions (\ref{occ1})-(\ref{occ4}), and it is an open problem whether the theory contains local common causes for all the correlations that are in need of a common cause explanation \cite{Redei-Summers2002}, \cite{Redei-Summers2005corr}, \cite{Hofer-Redei-Szabo-PCC}[Chapter 8]. But if a quantum probability theory is not common cause closed, then one cannot expect the theory to have common causes of the correlations that would need to be explained this way. This happens in discrete (lattice) quantum field theory of finite degrees of freedom. In these theories there exist correlations between observables pertaining to spacelike separated local algebras for which there exist no local common cause in the theory simply because these probability theories have many atoms and thus they contain a lot of correlations for which there are no common causes \emph{at all} in the theory \cite{Hofer-Vecsernyes2012a}, \cite{Hofer-Redei-Szabo-PCC}[Chapter 8].

\section{Some open questions\label{sec:open}}
The results in sections \ref{sec:sufficient} and \ref{sec:necessary} give a complete characterization of common cause closedness of probability theories understood as projection lattices of both commutative and non-commutative von Neumann algebras. Thus we have a characterization of common cause closedness of classical and quantum probability theories. As the propositions show, the same structural property is equivalent to common cause closedness in classical and quantum probability theories (not having more than one measure theoretic atom). It remains open however whether this structural property is equivalent to common cause closedness of general probability theories as well. Specifically, three open problems can be formulated.

The first problem is related to a sufficient condition of common cause closedness of general probability theories. We have seen that $Q$-decomposability plays an important role in common cause closedness. Lemma \ref{one_atom_Boolean} and Lemma \ref{one_atom_algebra} about $Q$-decomposability lead to the following problem:

\begin{problem}\label{prob:sufficient}
{\rm
Let $\mathcal{L}$ be a $\sigma$-complete orthomodular lattice and let $\phi$ be a faithful $\sigma$-additive probability measure on $\mathcal{L}$. Does the assumption that  $\mathcal{L}$ has at most one $\phi$-atom $Q$ entail that $\phi$ is $Q$-decomposable?
}
\end{problem}

The second problem is related to a necessary condition for common cause closedness in a general orthomodular lattice. We get a characterization of common cause closedness by using Lemma \ref{correlation_Boolean} and Lemma \ref{correlation_algebra}. These lemmas lead to the following problem:

\begin{problem}\label{prob:necessary}
{\rm
Let $\mathcal{L}$ be a $\sigma$-complete orthomodular lattice which has two distinct atoms and let $\phi$ be a faithful $\sigma$-additive probability measure on $\mathcal{L}$. Are the following two conditions equivalent?
\begin{enumerate}
\item There are two distinct elements $A$ and $B$ in $\mathcal{L}$ such that $\phi(A \wedge B)>\phi(A)\phi(B)$.
\item There are two distinct $\phi$-atoms $P$ and $Q$ such that $\phi(P \vee Q)<1$.
\end{enumerate}
If these conditions are not equivalent, is there a natural condition which is equivalent to Condition 2?
}
\end{problem}


The third problem is related to common cause completability. As mentioned in section \ref{sec:discussion}, common cause completability makes falsification of the Common Cause Principle difficult because it allows evading falsification by referring to hidden common causes. Not common cause closed classical probability spaces are common cause completable \cite{GyenisZ-Redei2011} (also see \cite{Hofer-Redei-Szabo-PCC}[Chapter 6] and \cite{Marczyk-Wronski2013a}). By the characterization of common cause closedness (Theorem \ref{characterization_algebra}), we can specify the quantum probability measure spaces which are not common cause closed. For example, the orthomodular lattice of all projections on a Hilbert space (Hilbert lattices) are not common cause closed because these probability theories contain a lot of measure theoretic atoms. Then the following question arises; is there a larger probability measure space into which these non-common cause closed quantum probability spaces are embeddable? The answer to this question is unknown. More generally we have

\begin{problem}
Can a general probability space $(\cL,\phi)$ be embedded into a larger one $(\cL',\phi')$ which has at most one $\phi$-atom?
\end{problem}
By an embedding is meant here an injective ortholattice homomorphism $h\colon\cL\to\cL'$ which preserves $\phi$ in the sense $\phi'(h(A))=\phi(A)$ for all $A\in\cL$.

Recently the notion of common cause has been given a general definition in terms of non-selective operations understood as completely positive, unit preserving linear maps on \C algebras \cite{Redei2014SHPMP}. That definition makes it possible to formulate common cause closedness of quantum field theories on possibly non-flat spacetimes in the categorial approach to relativistic quantum field theory. It would be interesting to find out what role the structural properties of the operator algebras play in common cause closedness (or lack of it) in such theories. Nothing is known about this question, so we omit the precise definitions.

\section*{Acknowledgments}

\noindent
Y. Kitajima is supported by the JSPS KAKENHI, No.23701009 and the John Templeton Foundation Grant ID 35771. \\
M. R\'edei acknowledges support of the Hungarian Scientific
Research Found (OTKA, contract number: K100715), and thanks the Institute of Philosophy of the Hungarian Academy of Sciences, with which he was affiliated as Honorary Research Fellow while this paper was written.

\bibliography{RedeiBib}

\begin{thebibliography}{10}

\bibitem{Araki1999}
H.~Araki.
\newblock {\em Mathematical Theory of Quantum Fields}, volume 101 of {\em
  International Series of Monograps in Physics}.
\newblock Oxford University Press, Oxford, 1999.
\newblock Originally published in Japanese by Iwanami Shoten Publishers, Tokyo,
  1993.

\bibitem{Cartwright1987}
N.~Cartwright.
\newblock How to tell a common cause: {G}eneralization of the conjunctive fork
  criterion.
\newblock In J.~H. Fetzer, editor, {\em Probability and Causality}, pages
  181--188. Reidel Pub. Co., 1987.

\bibitem{Fraassen1982b}
B.C.~Van Fraassen.
\newblock Rational belief and the common cause principle.
\newblock In R.~McLaughlin, editor, {\em What? Where? When? Why?}, pages
  193--209. D. Reidel Pub. Co., 1982.

\bibitem{Grasshoff-Portmann-Wuthrich2005}
G.S. Grasshoff, S.~Portmann, and A.~W\"uthrich.
\newblock Minimal assumption derivation of a {B}ell-type inequality.
\newblock {\em The British Journal for the Philosophy of Science}, 56:663--680,
  2005.

\bibitem{GyenisZ-Redei2011}
Z.~Gyenis and M.~R\'{e}dei.
\newblock Characterizing common cause closed probability spaces.
\newblock {\em Philosophy of Science}, 78:393--409, 2011.

\bibitem{GyenisZ-Redei2013Erkenntnis}
Z.~Gyenis and M.~R\'{e}dei.
\newblock Atomicity and causal completeness.
\newblock {\em Erkenntnis}, 79:437--451, 2014.
\newblock doi: 10.1007/s10670-013-9456-1.

\bibitem{Haag1992}
R.~Haag.
\newblock {\em Local Quantum Physics: Fields, Particles, Algebras}.
\newblock Springer Verlag, Berlin and New York, 1992.

\bibitem{Hofer1997}
G.~Hofer-Szab\'o.
\newblock The formal existence and uniqueness of the {R}eichenbachian common
  cause on {H}ilbert lattices.
\newblock {\em International Journal of Theoretical Physics}, 36:1973--1980,
  1997.

\bibitem{Hofer1998}
G.~Hofer-Szab\'o.
\newblock Reichenbach's common cause definition on {H}ilbert lattices.
\newblock {\em International Journal of Theoretical Physics}, 37:435--443,
  1997.

\bibitem{Hofer-Redei-Szabo1999}
G.~Hofer-Szab\'o, M.~R\'edei, and L.E. Szab\'o.
\newblock On {R}eichenbach's {C}ommon {C}ause {P}rinciple and {R}eichenbach's
  notion of common cause.
\newblock {\em The British Journal for the Philosophy of Science}, 50:377--398,
  1999.

\bibitem{Hofer-Redei-Szabo-PCC}
G.~Hofer-Szab\'o, M.~R\'edei, and L.E. Szab\'o.
\newblock {\em The Principle of the Common Cause}.
\newblock Cambridge University Press, Cambridge, 2013.

\bibitem{Hofer-Vecsernyes2012b}
G.~Hofer-Szab\'o and P.~Vecserny\'es.
\newblock Noncommutative local common causes for correlations violating the
  clauser-horne inequality.
\newblock {\em Journal of Mathematical Physics}, 53:122301, 2012.

\bibitem{Hofer-Vecsernyes2012a}
G.~Hofer-Szab\'o and P.~Vecserny\'es.
\newblock Reichenbach's common cause principle in algebraic quantum field
  theory with locally finite degrees of freedom.
\newblock {\em Foundations of Physics}, 42:241--255, 2012.

\bibitem{Johnson1970}
Roy~A. Johnson.
\newblock Atomic and nonatomic measures.
\newblock {\em Proceedings of the American Mathematical Society}, 25:650--655,
  1070.

\bibitem{Kadison-Ringrose1986}
R.V. Kadison and J.R. Ringrose.
\newblock {\em Fundamentals of the Theory of Operator Algebras}, volume I. and
  II.
\newblock Academic Press, Orlando, 1986.

\bibitem{Kalmbach1983}
G.~Kalmbach.
\newblock {\em Orthomodular Lattices}.
\newblock Academic Press, London, 1983.

\bibitem{Kitajima2008}
Y.~Kitajima.
\newblock Reichenbach's {C}ommon {C}ause in an atomless and complete
  orthomodular lattice.
\newblock {\em International Journal of Theoretical Physics}, 47:511--519,
  2008.

\bibitem{Marczyk-Wronski2013a}
M.~Marczyk and L.~Wronski.
\newblock A completion of the causal completability problem.
\newblock {\em The British Journal for the Philosophy of Science}, 2014.
\newblock Forthcoming, First published online: March 21, 2014, doi:
  10.1093/bjps/axt030.

\bibitem{Popper1959}
K.~Popper.
\newblock {\em The Logic of Scientific Discovery}.
\newblock Routledge, London and New York, 1995.
\newblock First published in English in 1959 by Hutchinson Education.

\bibitem{Portmann-Wuthrich2007}
S.~Portmann and A.~W\"uthrich.
\newblock Minimal assumption derivation of a weak {C}lauser--{H}orne
  inequality.
\newblock {\em Studies in History and Philosophy of Modern Physics},
  38:844--862, 2007.

\bibitem{Redei1998}
M.~R{\'e}dei.
\newblock {\em Quantum Logic in Algebraic Approach}, volume~91 of {\em
  Fundamental Theories of Physics}.
\newblock Kluwer Academic Publisher, 1998.

\bibitem{Redei2014CCPAssessment}
M.~R\'edei.
\newblock Assessing the status of the {C}ommon {C}ause {P}rinciple.
\newblock In M-C. Galavotti, D.~Dieks, W.J. Gonzalez, S.~Hartmann, T.~Uebel,
  and M.~Weber, editors, {\em New Directions in the Philosophy of Science},
  pages 433--442. Springer, Wien and New York, 2014.

\bibitem{Redei2014SHPMP}
M.~R\'edei.
\newblock A categorial approach to relativistic locality.
\newblock {\em Studies in Histroy and Philosophy of Modern Physics},
  48:137--146, 2014.

\bibitem{Redei-Summers2002}
M.~R\'edei and S.J. Summers.
\newblock Local primitive causality and the common cause principle in quantum
  field theory.
\newblock {\em Foundations of Physics}, 32:335--355, 2002.

\bibitem{Redei-Summers2007}
M.~R\'edei and S.J. Summers.
\newblock Quantum probability theory.
\newblock {\em Studies in the History and Philosophy of Modern Physics},
  38:390--417, 2007.

\bibitem{Redei-Summers2005corr}
M.~R\'edei and S.J. Summers.
\newblock Remarks on causality in relativistic quantum field theory.
\newblock {\em International Journal of Theoretical Physics}, 46:2053--2062,
  2007.

\bibitem{Reichenbach1956}
H.~Reichenbach.
\newblock {\em The Direction of Time}.
\newblock University of California Press, Los Angeles, 1956.

\bibitem{Salmon1978}
W.C. Salmon.
\newblock Why ask ``why?''?
\newblock In {\em Proceedings and Addresses of the American Philosophical
  Association}, volume~51, pages 683--705, 1978.

\bibitem{Salmon1984}
W.C. Salmon.
\newblock {\em Scientific Explanation and the Causal Structure of the World}.
\newblock Princeton University Press, Princeton, 1984.

\bibitem{Sober1988}
E.~Sober.
\newblock The principle of the common cause.
\newblock In J.H. Fetzer, editor, {\em Probability and Causality}, pages
  211--228. Reidel Pub. Co., Boston, 1988.

\bibitem{Sober1989}
E.~Sober.
\newblock Independent evidence about a common cause.
\newblock {\em Philosophy of Science}, 56:275--287, 1989.

\bibitem{Sober2001}
E.~Sober.
\newblock Venetian sea levels, {B}ritish bread prices, and the principle of
  common cause.
\newblock {\em The British Journal for the Philosophy of Science}, 52:331--346,
  2001.

\bibitem{Sober-Steel2012}
E.~Sober and M.~Steel.
\newblock Screening-off and causal incompleteness -- a {N}o-{G}o theorem.
\newblock {\em The British Journal for the Philosophy of Science}, 64:513--550,
  2013.

\bibitem{Wronski2014}
L.~Wronski.
\newblock {\em Reichenbach's Paradise. Constructing the Realm of Probabilistic
  Common ``Causes"}.
\newblock De Gruyter, Warsaw, Berlin, 2014.

\bibitem{Wuethrich2004}
A.~W\"uthrich.
\newblock {\em Quantum Correlations and Common Causes}.
\newblock Bern Studies in the History and Philosophy of Science. Universit\"at
  Bern, Bern, 2004.

\end{thebibliography}
\end{document}